\documentclass[10pt, sigconf]{acmart}
\usepackage{booktabs} 
\setcopyright{rightsretained}
\usepackage{graphicx}
\usepackage{subfigure}
\usepackage{amsfonts}
 
\usepackage{amssymb}
\usepackage{color}
\usepackage{epstopdf}
\usepackage{amsmath}
\usepackage{amsthm}
\usepackage{lipsum}
\usepackage{balance}
\usepackage{epsfig,latexsym}
\usepackage{bm}
\usepackage{mathrsfs}
\usepackage{rotating}
\usepackage{fancyhdr}
\usepackage[utf8]{inputenc}

\begin{document}
\settopmatter{printacmref=false}
\renewcommand\footnotetextcopyrightpermission[1]{}
\captionsetup[figure]{name={Fig.},labelsep=period,singlelinecheck=off}

\title{Rethinking the Performance of ISAC System: From Efficiency and Utility Perspectives}

\author{Jiamo Jiang$^{\dag}$, Mingfeng Xu$^{\dag}$, Zhongyuan Zhao$^{*}$, Kaifeng Han$^{\dag}$, Yang Li$^{\dag}$, Ying Du$^{\dag}$, and Zhiqin Wang$^{\dag}$}
\affiliation{
\institution{${\dag}$ China Academy of Information and Communications Technology.}
\institution{${*}$ Beijing University of Posts and Telecommunications.}
\city{Beijing}
\country{China}}
\email{jiangjiamo@caict.ac.cn}

\begin{abstract}
Integrated sensing and communications (ISAC) is an essential technology for the 6G communication system, which enables the conventional wireless communication network capable of sensing targets around. The shared use of pilots is a promising strategy to achieve ISAC. It brings a trade-off between communication and sensing, which is still unclear under the imperfect channel estimation condition. To provide some insights, the trade-off between ergodic capacity with imperfect channel estimation and ergodic Cramer-Rao bound (CRB) of range sensing is investigated. Firstly, the closed-form expressions of ergodic capacity and ergodic range CRB are derived, which are associated with the number of pilots. Secondly, two novel metrics named efficiency and utility are firstly proposed to evaluate the joint performance of capacity and range sensing error. Specifically, efficiency is used to evaluate the achievable capacity per unit of the sensing error, and utility is designed to evaluate the utilization degree of ISAC. Moreover, an algorithm of pilot length optimization is designed to achieve the best efficiency. Finally, simulation results are given to verify the accuracy of analytical results, and provide some insights on designing the slot structure.
\end{abstract}

\maketitle

\section{Introduction}
As the raise of new intelligent services, such as extended reality (XR) and holographic awareness, it is challenging to guarantee the quality of service via the communication system that only have the ability of communication \cite{b10}. ISAC is emerged as a promising technology to solve this problem, which can provide sensing ability as a secondary function via sharing the wireless resources and hardware \cite{b11}.

Currently, communication-centric design is regarded as one of three key fields in the research on ISAC \cite{b8}. In particular, sensing is designed as an extensional function under the condition that the basic communication requirements should be guaranteed. In \cite{b3}, a pilots multiplexed based scheme has been proposed, where the pilots are used for estimating channels and sensing remote targets simultaneously. Then the trade-off between the error rate and range CRB has been studied. Moreover, a virtual pilots generated scheme has been proposed in \cite{b6} for enhancing the velocity sensing performance with a small reduction of data rate. In \cite{b9}, the possibility of employing multiple slots in OFDM system for target detection has been investigated. Furthermore, the impact of imperfect self-interference cancellation at receiver has been considered in \cite{b7}, and a synchronized scheme between received signal and self-interference has been designed to employ multiple slots to detect multiple targets.

\begin{figure}
	\centering
	\includegraphics[width=0.8\linewidth]{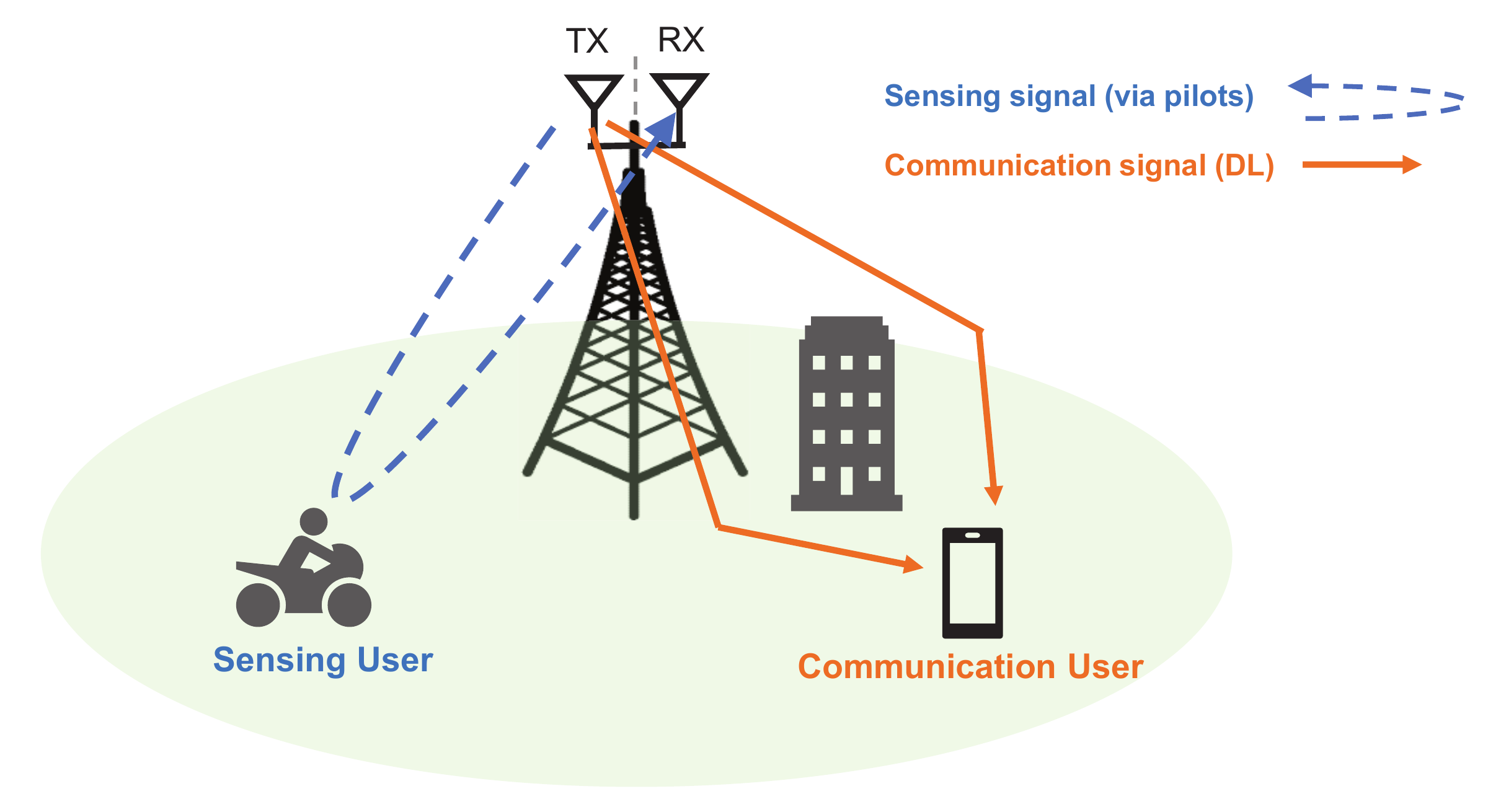}
	\caption{System model of ISAC network: An example.}
\end{figure}

All above works only consider the ideal channel conditions, which is not in line with the reality. Therefore, we study the trade-off between communications and sensing with consideration of imperfect channel estimation in this paper. Moreover, consider that the separation of communication metrics and sensing metrics can not characterize the efficiency and utility of ISAC, two novel metrics are firstly proposed in this paper. The main contributions can be summarized as follows:
\begin{itemize}
\item First, the performance of ergodic capacity with imperfect channel estimation and ergodic range CRB are studied, and their closed-form expressions related to the number of pilots are provided. Followed by establishing an exact trade-off between capacity and range CRB.
\item Second, two novel metrics aiming to evaluate the joint performance of capacity and range sensing error are proposed, named efficiency and utility, respectively. In particular, efficiency is to evaluate achievable capacity of unit sensing error, and utility is to evaluate the reachability of maximum capacity and minimum sensing error. Moreover, the best efficiency and utility can be achieved by optimizing the number of pilots.
\item Finally, simulation results are provided to show how the number of pilots affects the joint performance of communications and sensing, which can provide some insights on designing the slot structure.
\end{itemize}

\section{System Model.}
Consider a downlink transmission scenario where the base station (BS) equiped with a pair of antennas serves two typical users, as shown in Fig. 1. Typically, User 1 is a communication target equiped with a single antenna, while User 2 is a device-free sensing target. The BS can determine its distance to User 2 via receiving the echo of a series of pilot signals without declining the communication service quality for User 1. To support the measurement capable of high precision, a line-of-sight (LOS) path is assumed to be existence in the signal propogation between the BS and User 2. Moreover, the non-line-of-sight (NLOS) environment is considered between the BS and User 1. Then, the distribution of wireless channel between the BS and User 1 is modeled as Rayleigh distribution, while the wireless channel between the BS and User 2 follows Rice distribution.

As shown in Fig. 2, assuming that a transmission slot contains $L$  symbols. In particular, two phases are existed in each transmission slot, which are named the pilot transmission phase and the data transmission phase, respectively. Without loss of generality, the pilot signals are assumed to be transmitted at the first $L_p$ symbols, and the rest $L-L_p$ symbols belong to the data transmission phase. Then, the received signals at User 1 can be expressed as:

\begin{subequations}\label{eqn: received_ue1}
\begin{align}
& y_{p}(t) = h_{p}(t) \sqrt{\rho_p} x_{p}(t) + w_p(t),\ 1 \leq t \leq L_p,\\
& y_{d}(t) = h_{d}(t) \sqrt{\rho_d} x_{d}(t) + w_d(t),\ L_p+1 \leq t \leq L,
\label{eqn: received_ue1_hd}
\end{align}
\end{subequations}
where $y_{p}$ and $y_{d}$ are received pilot and data signals at User 1, respectively, $x_{p}$ and $x_{d}$ are the transmitted pilot and data symbols, respectively, $h_{p}$ and $h_{d}$ represent the Rayleigh fading coefficient of the link between the BS and User 1 in the pilot and data transmission phase, respectively, $\rho_p$ and $\rho_d$ are the transmitted power of pilot and data symbols, respectively, and $w_p$ and $w_d$ are the Gaussian noise variable with variance $\sigma_p^2$ and $\sigma_d^2$, respectively.

\begin{figure}
	\centering
	\includegraphics[width=0.8\linewidth]{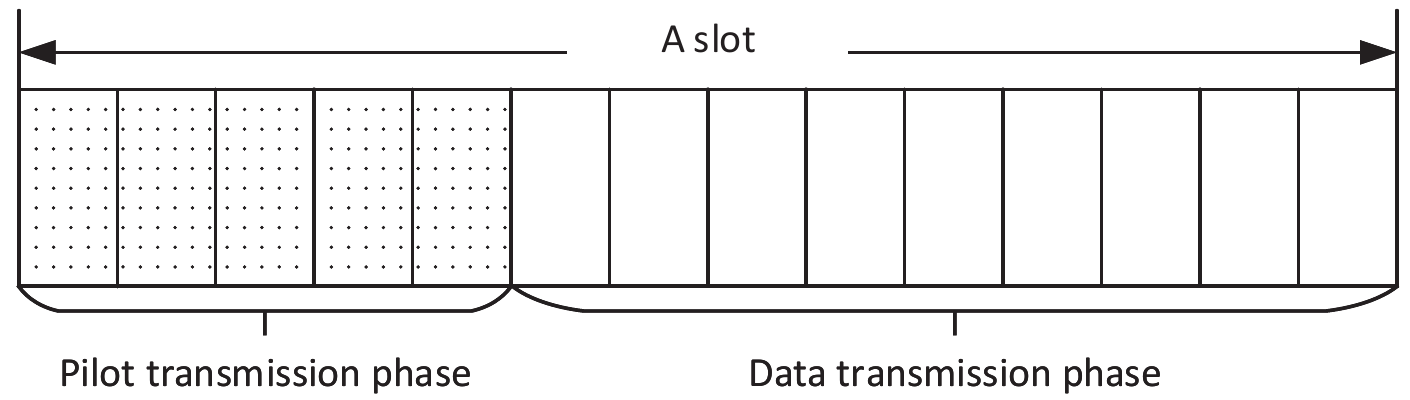}
	\caption{The structure of a slot.}
\end{figure}

Consider that only pilot signals are reused as the probe signal for sensing User 2. Besides, assuming that the transmitted signals can be cancelled successfully by the receiver at source such that the self-interfence is ignored. Then, the echo that reflected by User 2 received at the BS can be expressed as:
\begin{equation}\label{eqn:received_bs}
y_{s}(t) = \sqrt{\rho_{p,r} s_{rcs}} h_s(t)  x_p\bigg(t-\frac{2d}{c}\bigg) e^{\frac{j 4\pi vt}{\lambda}} + w_s(t),
\end{equation}
where $\rho_{p,r}=\rho_p / (4\pi d^2)$ denotes the received echo strength at BS, $h_s$ denotes the Rice fading coefficient of the link between the BS and User 2, $s_{rcs}$ is the radar cross section (RCS) of User 2, $d$ is the range between the BS and User 2, $c$ is the light of speed, $v$ is the speed of User 2 moving towards the BS, $\lambda$ is the wavelength of of the carrier, and $w_s$ is the additive white Gaussian noise with variance $\sigma_s^2$.

\section{Communication Performance: Ergodic Capacity with Imperfect Channel Estimation.} \label{sec:capacity}
In this section, we investigate the impact of imperfect channel estimation on the performance of ergodic capacity, and a closed-form expression of ergodic capacity associated with the number of pilots is provided. Firstly, with the estimated results of pilot channels, the data channels are estimated by employing wiener filtering interpolation algorithm. Then the ergodic capacity is derived with data channel estimation error, which establishes a relationship between the capacity and the number of pilots.

\subsection{Channel Estimation via Pilots.}
At first, the pilot channels need to be determined, which can be solved quickly by using the least-square estimation method. Thus, the expected error of estimation for pilot channels can be written as $e_p = \sigma_p^2 / \rho_p$. Next, to achieve the minimum mean square error (MMSE) criteria, the wiener filtering interpolation algorithm \cite{b1} is employed. Then the error of estimation can be given by the following lemma.

\begin{lemma}
The expected error of estimation for data channels can be expressed as:
\begin{equation}\label{eqn:ed_exp}
e_d = \sigma_1^2 - \frac{\sigma_1^2}{1+\frac{1}{ \sigma_1^2 \gamma_p L_p}},
\end{equation}
where $\sigma_1^2$ is the variance of the Rayleigh fading channel, $\gamma_p=\rho_p/\sigma_p^2$ denotes the signal-noise-ratio (SNR) of pilot channels.
\end{lemma}
\begin{proof}
Due to the limitation of paper space, the detailed proof of lemma 3.1 is omitted.
\end{proof}

As shown in Lemma 3.1, the expected error $e_d$ is a decreasing function with respect to $L_p$, which means that the estimation performance will be better as the number of pilots increases. $e_d$ will be $0$ when $L_p$ approaches to infinity.

\subsection{Ergodic Capacity with Imperfect Channel Estimation.}
Based on the expression of $e_d$ given by Lemma 3.1, the ergodic capacity for User 1 can be investigated. By using the reparametrization trick \cite{b12}, the Rayleigh fading coefficient of data channels $h_d$ can be rewritten as:
\begin{equation}\label{eqn:hd_re}
h_d = \hat{h}_d + w_{d,wf},
\end{equation}
where $\hat{h}_d$ denotes the estimated result of data channels, and $w_{d,wf}$ is the Gaussian variable with mean 0 and variance $e_d$. Then, by substituting \eqref{eqn:hd_re} into \eqref{eqn: received_ue1_hd}, the received signals at User 1 in the data transmission phase can be rewritten as:
\begin{equation}
\begin{aligned}
y_d(t) &= (\hat{h}_d (t) + w_{d,wf}(t)) \sqrt{\rho_d} x_d(t) + w_d(t)\\
&= \hat{h}_d (t) \sqrt{\rho_d} x_d(t) + w_{d,wf}(t) \sqrt{\rho_d} x_d(t) + w_d(t).
\end{aligned}
\end{equation}

As introduced in \cite{b2}, the ergodic capacity with channel estimation error for User 1 can be determined by
\begin{equation}\label{eqn:c}
C = \mathbb{E} \bigg\{ B(L-L_p) \log_2 \bigg( 1+\frac{|\hat{h}_d|^2 \rho_d}{\sigma_d^2 + |w_{d,wf}|^2 \rho_d} \bigg) \bigg\}.
\end{equation}
where $B$ is the bandwidth, and $|\cdot|$ denotes the norm operation.

Based on \eqref{eqn:c}, an closed-form expression of $C$ can be given in the following theorem.
\begin{theorem}
The ergodic capacity $C$ with imperfect channel estimation can be expressed as:
\begin{equation}\label{eqn:ergodic_capacity_exp}
\begin{aligned}
C = \frac{B(L-L_p)}{ \big(1-\frac{1}{1+\sigma_1^2 \gamma_p L_p} \big) \ln2} \bigg[ & e^{\frac{1 + \gamma_p L_p}{\gamma_d}}{\rm{Ei}} \bigg(-\frac{1 + \gamma_p L_p}{\gamma_d} \bigg) -\\
& e^{\frac{1}{\gamma_d \sigma_1^2}} {\rm{Ei}} \bigg(-\frac{1}{\gamma_d \sigma_1^2} \bigg) \bigg],
\end{aligned}
\end{equation}
where $\gamma_d=\rho_d/\sigma_d^2$ denotes the SNR of data channels, and ${\rm{Ei}}(x) = \int_{-\infty}^{x} e^t/t {\rm{d}} t$ is the exponential integral function.
\end{theorem}
\begin{proof}
The integral form of $\ln(1+x)$ is used instead, which can be expressed as:
\begin{equation}\label{eqn:ln_expand}
\ln(1+x) = \int_0^{\infty} \frac{e^{-z}}{z} (1-e^{-xz}) {\rm{d}} z.
\end{equation}
By substituting \eqref{eqn:ln_expand} into \eqref{eqn:c}, $C$ in \eqref{eqn:c} can be rewritten as:
\begin{equation}
\begin{aligned}
C = \frac{B(L-L_p)}{\ln2} \mathbb{E} \bigg\{ & \int_0^{\infty} \frac{e^{-z}}{z} \bigg[ 1 - \\
&\exp \bigg( -\frac{|\hat{h}_d|^2 \rho_d}{\sigma_d^2 + |w_{d,wf}|^2 \rho_d}z \bigg) \bigg] {\rm{d}} z \bigg\}.
\end{aligned}
\end{equation}
By employing Fubini's theorem, the order of the expectation operation and the intergration operation can be exchanged. Besides, for simplicity, $z$ is substituted by $(\sigma_d^2 + |w_{d,wf}|^2 \rho_d)s$, then $C$ can be further expressed as:
\begin{equation}\label{eqn:C_s}
\begin{aligned}
C = \frac{B(L-L_p)}{\ln2} \int_0^{\infty} & \frac{e^{-\sigma_d^2 s}}{s} \underbrace{\mathbb{E} \big\{ 1- e^{-|\hat{h}_d|^2 \rho_d s} \big\}}_{\mathcal{I}} \times \\
& \underbrace{\mathbb{E} \big\{ e^{-|w_{d,wf}|^2 \rho_d s} \big\}}_{\mathcal{J}} {\rm{d}} s.
\end{aligned}
\end{equation}

Since $\hat{h}_d \sim \mathcal{CN} (0,\sigma_1^2)$ and $w_{d,wf} \sim \mathcal{CN} (0,e_d)$, their probability density function (PDF) can be expressed as follows, respectively.
\begin{subequations}
\begin{align}
f_{|\hat{h}_d|^2}(x) &= \frac{1}{\sigma_1^2} \exp(-\frac{x}{\sigma_1^2}), \label{eqn:pdf_hd}\\
f_{|w_{d,wf}|^2}(x) &= \frac{1}{e_d} \exp \bigg( -\frac{x}{e_d} \bigg). \label{eqn:pdf_w} 
\end{align}
\end{subequations}
By substituting \eqref{eqn:pdf_hd} and \eqref{eqn:pdf_w} into $\mathcal{I}$ and $\mathcal{J}$ in \eqref{eqn:C_s}, respectively, $\mathcal{I}$ and $\mathcal{J}$ can be derived as follows, respectively.
\begin{subequations}\label{eqn:I+J}
\begin{align}
\mathcal{I} &= \frac{1}{\sigma_1^2} \int_0^{\infty} (1-e^{-\rho_d s x}) e^{-\frac{x}{\sigma_1^2}} {\rm{d}} x = \frac{\rho_d \sigma_1^2 s}{1+\rho_d \sigma_1^2 s},\\
\mathcal{J} &= \frac{1}{e_d} \int_0^{\infty} e^{-\rho_d s x} e^{-\frac{x}{e_d}} {\rm{d}} x = \frac{1}{1+\rho_d e_d s}.
\end{align}
\end{subequations}

Then, by substituting \eqref{eqn:I+J} into \eqref{eqn:C_s}, $C$ can be further derived as:
\begin{equation}\label{eqn:C_derive}
\begin{aligned}
C &= r \rho_d \sigma_1^2 \int_0^{\infty} \frac{e^{-\sigma_d^2 s}}{(1+\rho_d \sigma_1^2 s)(1+\rho_d e_d s)} {\rm{d}} s\\
& \overset{(a)}{=} \frac{r}{\rho_d e_d} \int_0^{\infty} \frac{e^{-\sigma_d^2 s}}{[1/(\rho_d \sigma_1^2)+s][1/(\rho_d e_d) +s]} {\rm{d}} s\\
& \overset{(b)}{=} \frac{r}{1-e_d/\sigma_1^2} \int_0^{\infty} \bigg( \frac{e^{-\sigma_d^2 s}}{1/(\rho_d \sigma_1^2)+s} - \frac{e^{-\sigma_d^2 s}}{1/(\rho_d e_d) +s} \bigg) {\rm{d}} s\\
&= \frac{r}{1-\frac{e_d}{\sigma_1^2}} \bigg[ e^{\frac{1}{\gamma_d e_d}}{\rm{Ei}} \bigg(-\frac{1}{\gamma_d e_d} \bigg) - e^{\frac{1}{\gamma_d \sigma_1^2}} {\rm{Ei}} \bigg(-\frac{1}{\gamma_d \sigma_1^2} \bigg) \bigg]. 
\end{aligned}
\end{equation}
where $r=B(L-L_p)/(\ln2)$, $(a)$ in the second equality can be obtained by extracting $1/(\rho_d^2 e_d)$ from the denominator of the integral term, and $(b)$ in the third equality can be obtained by splitting the denominator of the integral term. Finally, \eqref{eqn:ergodic_capacity_exp} can be obtained by substituting \eqref{eqn:ed_exp} into \eqref{eqn:C_derive}. Then, the proof has been finished.
\end{proof}

As shown in Theorem 3.2, the relationship between $C$ and $L_p$ is not clear. Therefore, the following derivatives of $C$ with respect to $L_p$ are given in \eqref{eqn:C_deri}. Since the first-order derivative can be positive or negative, and the second-order derivative is negative, $C$ has a maximum value. It means that a performance gain can be obtained at first when increasment of the number of pilot provides more precise channel estimation. However, the performance will be worse as the number of pilots increases continually due to the cost of symbols.
\begin{figure*}
\begin{equation}\label{eqn:C_deri}
\nabla_{L_p} ^2 C = -\frac{2B}{\ln2} \int_0^{\infty} \frac{\rho_d^2 \gamma_p e^{-\sigma_d^2 s} s}{(1+\rho_d s)(1+\gamma_p L_p + \rho_d s)^2} {\rm{d}} s - \frac{2(L-L_p)B}{\ln2} \int_0^{\infty} \frac{\rho_d^2 \gamma_p^2 e^{-\sigma_d^2 s} s}{(1+\rho_d s)(1+\rho_p L_p +\rho_d s)} {\rm{d}} s \leq 0. 
\end{equation}
{\noindent} \rule[-10pt]{18cm}{0.05em}
\end{figure*}

\section{Sensing Performance: Ergodic Range CRB.}\label{sec:CRB}
In this section, we investigate the performance of ergodic CRB for measuring the range of User 2, and provide its closed-form expression.

As introduced in \cite{b3}, the expression of CRB for range sensing can be expressed as:
\begin{equation}\label{eqn:CRB_exp}
CRB_d(L_p) = \frac{\alpha}{L_p |h_s(t)|^2},
\end{equation}
where $\alpha = c^2 / (8\pi^2 \gamma_{p,s} s_{rcs} B^2_{rms})$, $\gamma_{p,s}=\rho_p/\sigma_s^2$ is the SNR of echo channels, and $B_{rms}=\sqrt{ \int_{-\infty}^{\infty} F^2 |S(F)|^2 {\rm{d}} F / \int_{-\infty}^{\infty} |S(F)|^2 {\rm{d}} F}$ is the root-mean square bandwidth.

Based on \eqref{eqn:CRB_exp}, the closed-form expression of the expectation of $CRB_d$, i.e., $\bar{\delta}_d = \mathbb{E} \{CRB_d\}$, can be given by the following theorem.
\begin{theorem}
The ergodic CRB for measuring range $\bar{\delta}_d$ can be expressed as:
\begin{equation}\label{eqn:delta_d}
\bar{\delta}_d = \frac{\alpha \sqrt{\pi} \exp \big({-\frac{A_s}{2\sigma_2^2}}\big) \ _1F_1[1/2;1;A_s/(2\sigma_2^2)] }{L_p \sqrt{2\sigma_2^2}},
\end{equation}
where $\ _1F_1(\cdot)$ denotes the confluent hypergeometric function, $A_s$ denotes the signal strength of LOS path, $\sigma_2^2$ is the signal strength of the multipath.

\end{theorem}
\begin{proof}

Recalling that the wireless link between the BS and User 2 follows Rice distribution. It's PDF is expressed as:
\begin{equation}\label{eqn:pdf_rice}
f_{|h_s(t)|^2}(x) = \frac{x}{\sigma_2^2} \exp(-\frac{x^2+A_s}{2\sigma_2^2}) I_0 \bigg( \frac{\sqrt{A_s} x}{\sigma_2^2} \bigg),
\end{equation}
where $I_0(\cdot)$ is the zero-order Bessel function. Then, $\bar{\delta}_d = \mathbb{E}\{CRB_d\}$ can be derived as follows by substituting \eqref{eqn:pdf_rice} into \eqref{eqn:CRB_exp} and expanding $I_0(\cdot)$.
\begin{equation}\label{eqn:delta_d_deri1}
\bar{\delta}_d = \frac{\alpha e^{-\frac{A_s}{2\sigma_2^2}}}{L_p \sigma_2^2}  \sum_{k=0}^{\infty} \frac{1}{k! \Gamma(k+1)} \underbrace{\int_{0}^{\infty} e^{\frac{-x^2}{2\sigma_2^2}} \bigg( \frac{A_s x^2}{4\sigma_2^4} \bigg)^k {\rm{d}} x}_{\mathcal{K}},
\end{equation}
where $\Gamma(\cdot)$ is the Gamma function. Next, by substituting $x$ by $1/\sqrt{2\sigma_2^2} z$, $\mathcal{K}$ in \eqref{eqn:delta_d_deri1} can be further derived as:
\begin{equation}\label{eqn:delta_d_deri2}
\begin{aligned}
\mathcal{K} &= (2\sigma_2^2)^{k+\frac{1}{2}} \bigg( \frac{A_s}{4\sigma_2^4} \bigg)^k \int_{0}^{\infty} e^{-\frac{x^2}{2\sigma_2^2}} \bigg( \frac{z}{\sqrt{2\sigma_2^2}} \bigg)^{2k} {\rm{d}} \bigg(\frac{z}{\sqrt{2\sigma_2^2}} \bigg) \\
&= \frac{1}{2} (2\sigma_2^2)^{k+\frac{1}{2}} \bigg( \frac{A_s}{4\sigma_2^4} \bigg)^k \Gamma \bigg({k+\frac{1}{2}}\bigg).
\end{aligned}
\end{equation}
Finally, by substituting \eqref{eqn:delta_d_deri2} into \eqref{eqn:delta_d_deri1}, $\bar{\delta}_d$ in \eqref{eqn:delta_d} can be obtained. Then, the proof of Theorem 3.2 has been finished.
\end{proof}

\section{New metrics to characterize ISAC performance: Efficiency and utility.}\label{sec:Joint_capacity_CRB}

Recalling the ergodic capacity $C$ given by Theorem 3.2 and the ergodic CRB of range sensing shown in Theorem 4.1, they are both concerned with the number of pilots $L_p$. To make sense how $L_p$ balances the trade-off between the communication performance and the sensing performance, the details about that will be discussed in this section.

\subsection{Efficiency of ISAC}
In particular, a novel performace metric named the efficiency of ISAC is proposed, which aims to determine the optimal $L_p$ that can strike the balance between the capacity and the sensing error. It is defined as the ratio of capacity to CRB and can be written as:
\begin{equation}\label{eqn: Efficiency}
E_{{\rm{ISAC}}}^{\gamma}(L_p) = \frac{C^{\gamma}(L_p)}{\kappa + CRB_d^{\gamma}(L_p)},
\end{equation}
where $\kappa$ is a constant that can limit the maximum value of the efficiency, $\gamma$ denotes the SNR. $E_{{\rm{ISAC}}}^{\gamma}(L_p)$  evaluates the achievable capacity per unit of the sensing error. It indicates that the efficiency can be improved as $E_{{\rm{ISAC}}}^{\gamma}(L_p)$ becomes larger. Therefore, to reach the best efficiency, an optimization problem is formulated as
\begin{equation}
\mathcal{P}_1 : \max_{L_p} E_{{\rm{ISAC}}}^{\gamma} (L_p),~s.t.~1 \leq L_p \leq L-1.
\end{equation}

As $\mathcal{P}_1$ is a non-linear fractional problem, it can be solved by transforming the original problem to the following equivalent linear problem \cite{b4}
\begin{equation}
\mathcal{P}_2 : \max_{L_p} Q_2 = C^{\gamma}(L_p^*) - q^{*} (\kappa + CRB_d^{\gamma}(L_p^*) ),
\end{equation}
where $L_p^*$ is the optimal solution of $\mathcal{P}_1$, and $q^*$ denotes the maximum value of $\mathcal{P}_1$. By denoting $\bar \delta_d$ given by \eqref{eqn:delta_d} as $\beta(\gamma) L_p^{-1}$, the second-order derivative of $Q_2$ can be written as:
\begin{equation}\label{eqn:P4_deri}
\nabla_{L_p}^2 Q_2 = \nabla_{L_p}^2 C^{\gamma}(L_p) - 2 q^* \beta(\gamma) L_p^{-3} \leq 0,
\end{equation}
where $\nabla_{L_p}^2 $ has already been given by \eqref{eqn:C_deri}.

According to \eqref{eqn:P4_deri}, $Q_2$ is a concave function and its optimal solution can be determined by solving $\nabla_{L_p} Q_2=0$. By using Algorithm 1 provided in \cite{b4}, the optimal solution of original problem $\mathcal{P}_1$ can be obtained, which is expressed as:
\begin{equation}\label{eqn:Lp_opt}
L_p^{opt} = \left\{
\begin{aligned}
&1,~1 \geq L_p^*, \\
&L_p^*,~ 1 < L_p^* < L-1, \\
&L-1,~L_p^* \geq L-1.
\end{aligned}
\right.
\end{equation}
\subsection{Utility of ISAC}
Moreover, a novel concept named utility of ISAC is proposed, which is designed to evaluate the utilization degree of joint performance for communication and sensing. It contains two terms, the first term named capacity utility depicts the ratio of capacity with $L_p$ pilots to the maximum achievable capacity, and the second term named sensing utility depicts the ratio of CRB with $L_p$ pilots to the minimum achievable CRB. These two terms are integrated by an adjustable weighted factor, which is adaptable for specific services. Hence, the definition of utility can be given as:
\begin{equation}\label{eqn:utility}
U_{{\rm{ISAC}}}^{\gamma} = \eta \frac{C^{\gamma}(L_p)}{C_{\max}^{\gamma}} + (1-\eta) \frac{CRB_{g,\min}^{\gamma}}{CRB_g^{\gamma}(L_p)},
\end{equation}
where $C_{\max}^{\gamma}$ denotes the maximum capacity, $CRB_{g,\min}^{\gamma}$ denotes the minimum estimation error for sensing, and $\eta \in (0,1)$ is the weighted factor. 

Since the monotonicity of the capacity and CRB terms can be determined when parameters are fixed, the monotonicity of $U_{{\rm{ISAC}}}^{\gamma}$ is decided by $\eta$, which can be divided into the following two cases: $1)$ when $\eta$ is a small value, $U_{{\rm{ISAC}}}^{\gamma}$ is an increasing function as the CRB term plays a major role; $2)$ otherwise, $U_{{\rm{ISAC}}}^{\gamma}$ is a concave function owning a maximum value as the capacity term has more impact.

Therefore, to avoid the allocation of pilots to extremes, the threshold of capacity utility and sensing utility should be considered, which are denoted as $U_{c,th}$ and $U_{d,th}$, respectively. Only $L_p$ that satisfies the requirements is in consideration.

\section{Simulation Results.}
In this section, simulation results are provided to verify the derived results and evaluate the joint performance of communication and sensing. In particular, the total symbols transmitted in a slot is set as $L=14$, the bandwidth is set as $B=200 \ {\rm{MHz}}$, the variance of the link between BS and User 1 is set as $\sigma_1^2 = 2$, the Rician factor of the LOS path is set as $K = A_s/\sigma_2^2 = 3$. Moreover, the radar cross section is $s_{rcs}=100\ m^2$, and $B_{rms} = B/\sqrt{12}$ according to \cite{b5}.

In Fig. 3, the performance of the expection of capacity and range CRB are presented in (a) and (b), respectively. As shown in Fig. 3 (a), owing to the accuracy improvement of channel estimation, the capacity performance can be improved when $L_p$ begins to increase. However, total capacity decreases as $L_p$ keeps increasing, which results from the cost of resource block. Besides, channel estimation in low SNR case can provide larger performance gain than in high SNR case. In Fig. 3 (b), $\gamma_p=10$ dB is considered. As shown in this figure, the increasement of $L_p$ provides a more precise range estimation result. Moreover, all analytical results are matched perfectly with the Monte-Carlo results, which demonstrates the accuracy of our derived results.

\begin{figure}[htbp]\centering
\subfigure[Ergodic capacity.]{\centering 
\includegraphics[width=0.23\textwidth]{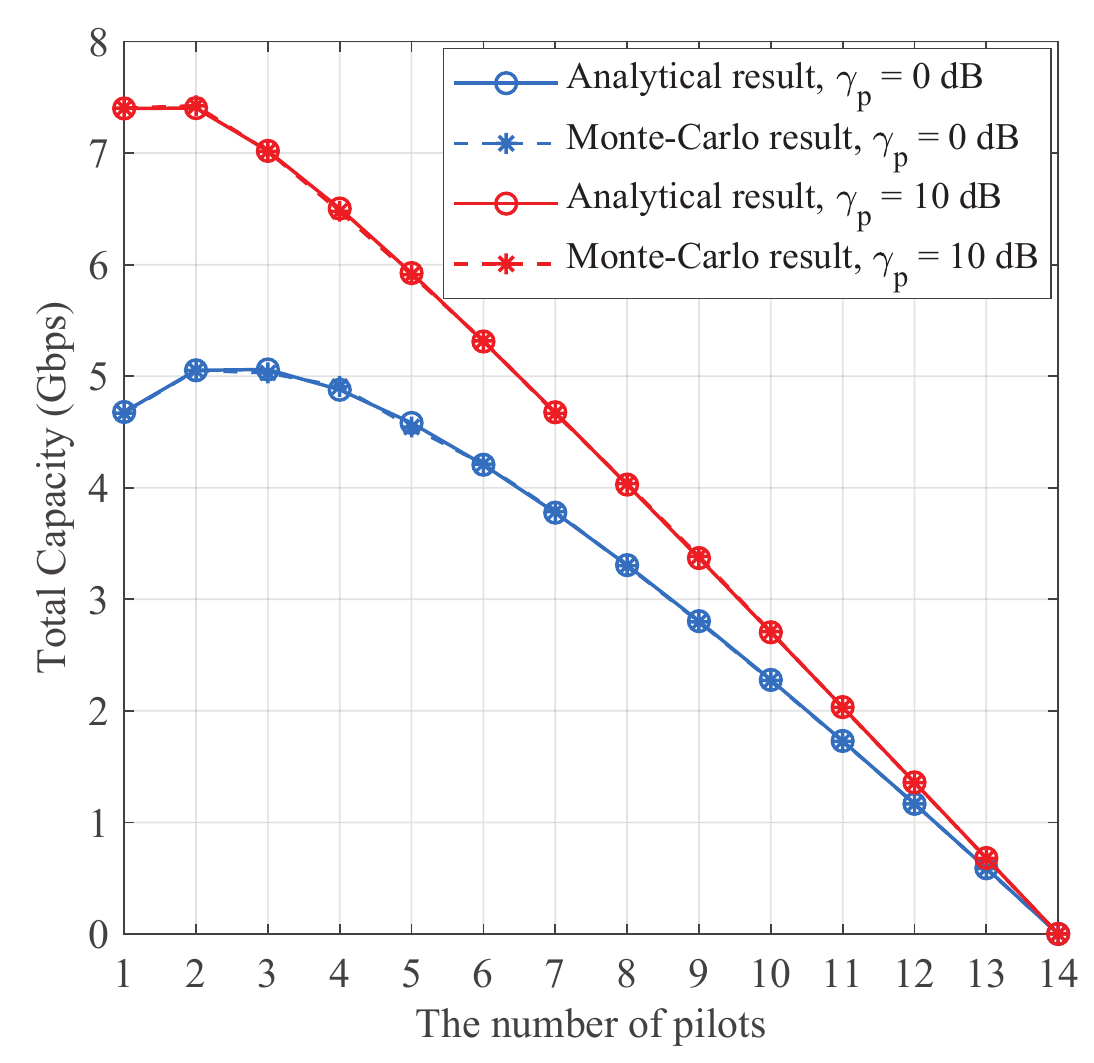}}
\subfigure[Ergodic range CRB.]{\centering 
\includegraphics[width=0.23\textwidth]{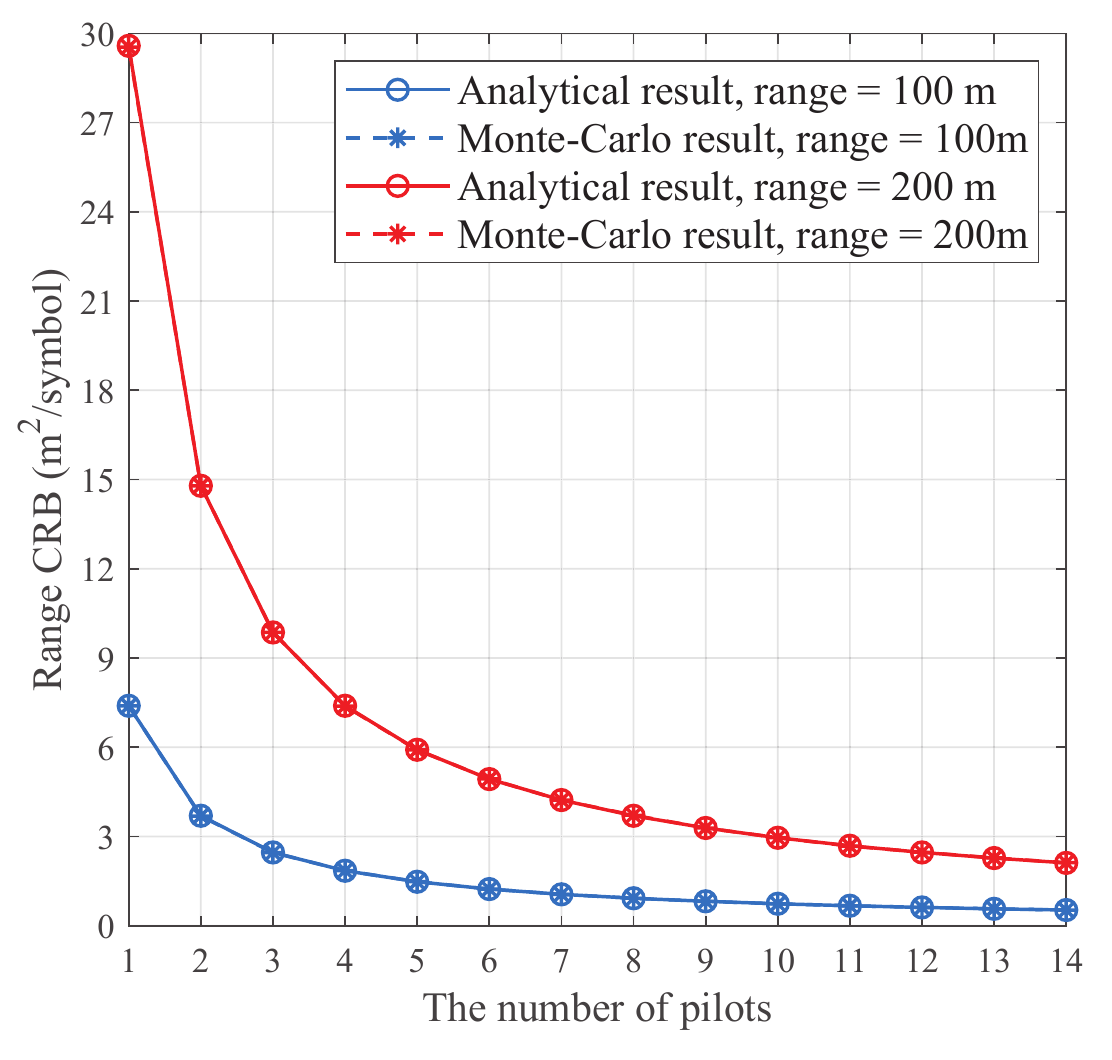}}
\caption{Performance of ergodic capacity and range CRB with respect to the number of pilots. }
\end{figure}

In Fig. 4, the trade-off between the capacity and range CRB with respect to the number of pilots is investigated. In particular, the measured range is set as $100$ m. As shown in this figure, three regions can be divided as follows. $1)$ the upper right region: this region with high capacity but low sensing accuracy; $2)$ the bottom left region: this region with high sensing accuracy but low capacity; $3)$ the bottom right region: this region with high capacity and high sensing accuracy. Therefore, to achieve a balance between communication and sensing, the bottom right region is recommended to be considered more in the ISAC design.

\begin{figure}
	\centering
	\includegraphics[width=0.6\linewidth,height=0.23\textwidth]{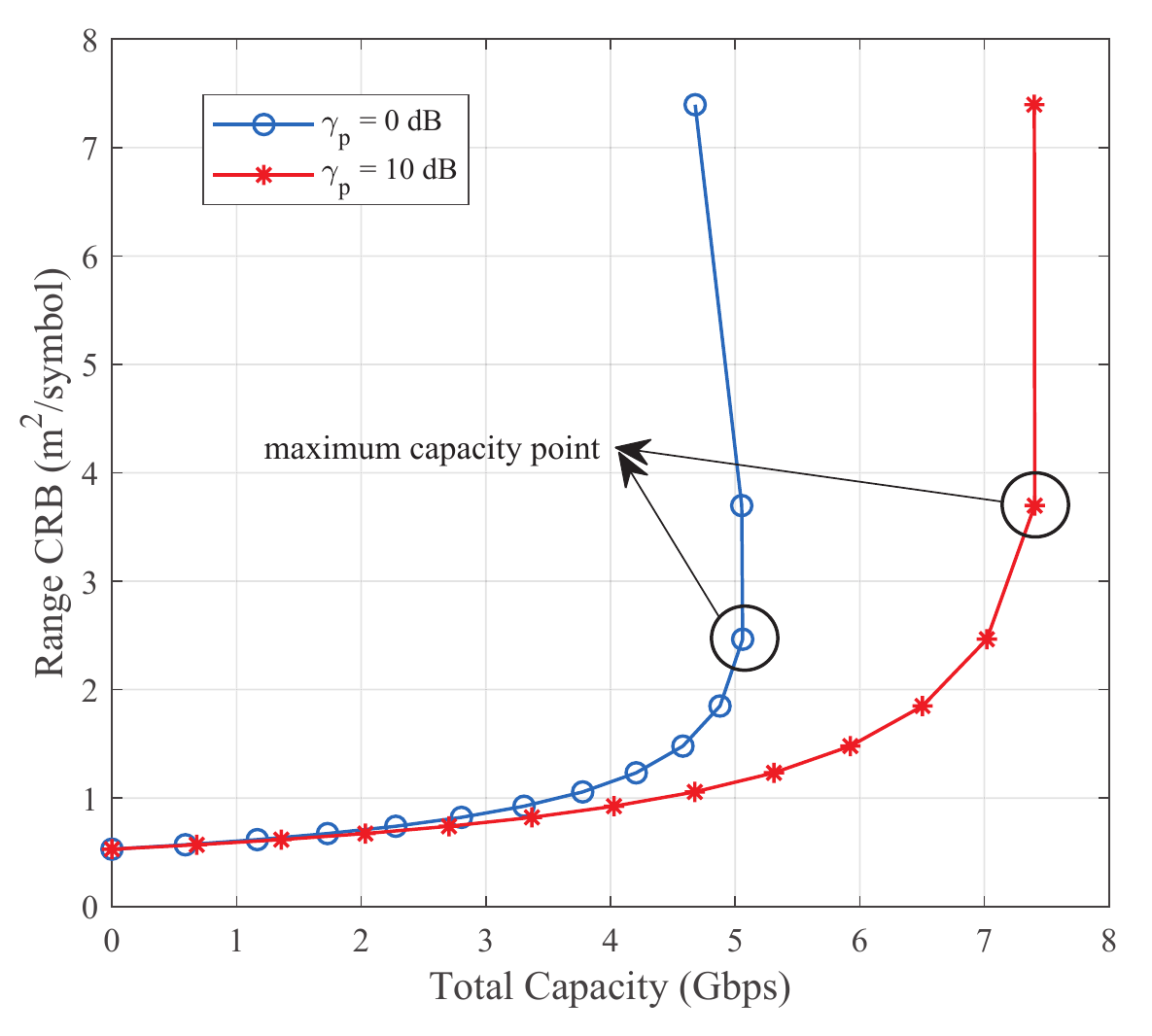}
	\caption{Performance trade-off between communications and range sensing with respect to the number of pilots.}
\end{figure}

In Fig. 5, the efficiency of ISAC proposed in \eqref{eqn: Efficiency} with $\kappa=1$ is evaluated. As shown in Fig. 5 (a), the efficiency with respect to $L_p$ is depicted. It can be observed that the peak of efficiency is also related to the SNR settings. Therefore, we provide the efficiency with respect to SNR under different $L_p$ settings in Fig. 5 (b). In particular, the optimal $L_p$ is obtained by \eqref{eqn:Lp_opt}, and the cases that $L_p = 2$, $L_p=7$ and $L_p=12$ is selected to be compared, which correspond to the upper right region, bottom right region and bottom left region, respectively. As shown in this figure, the case with optimal $L_p$ always achieves the best efficiency. In addtion, it is implied that a certain amount of pilots is needed to maintain the balance between capacity and range CRB in medium and low SNR scenario. However, the decrement of $L_p$ would lead to a higher efficiency in high SNR scenario, as the accuracy of range measurement is precise enough.

\begin{figure}[htbp]\centering
\subfigure[Efficiency with $L_p$.]{\centering 
\includegraphics[width=0.235\textwidth]{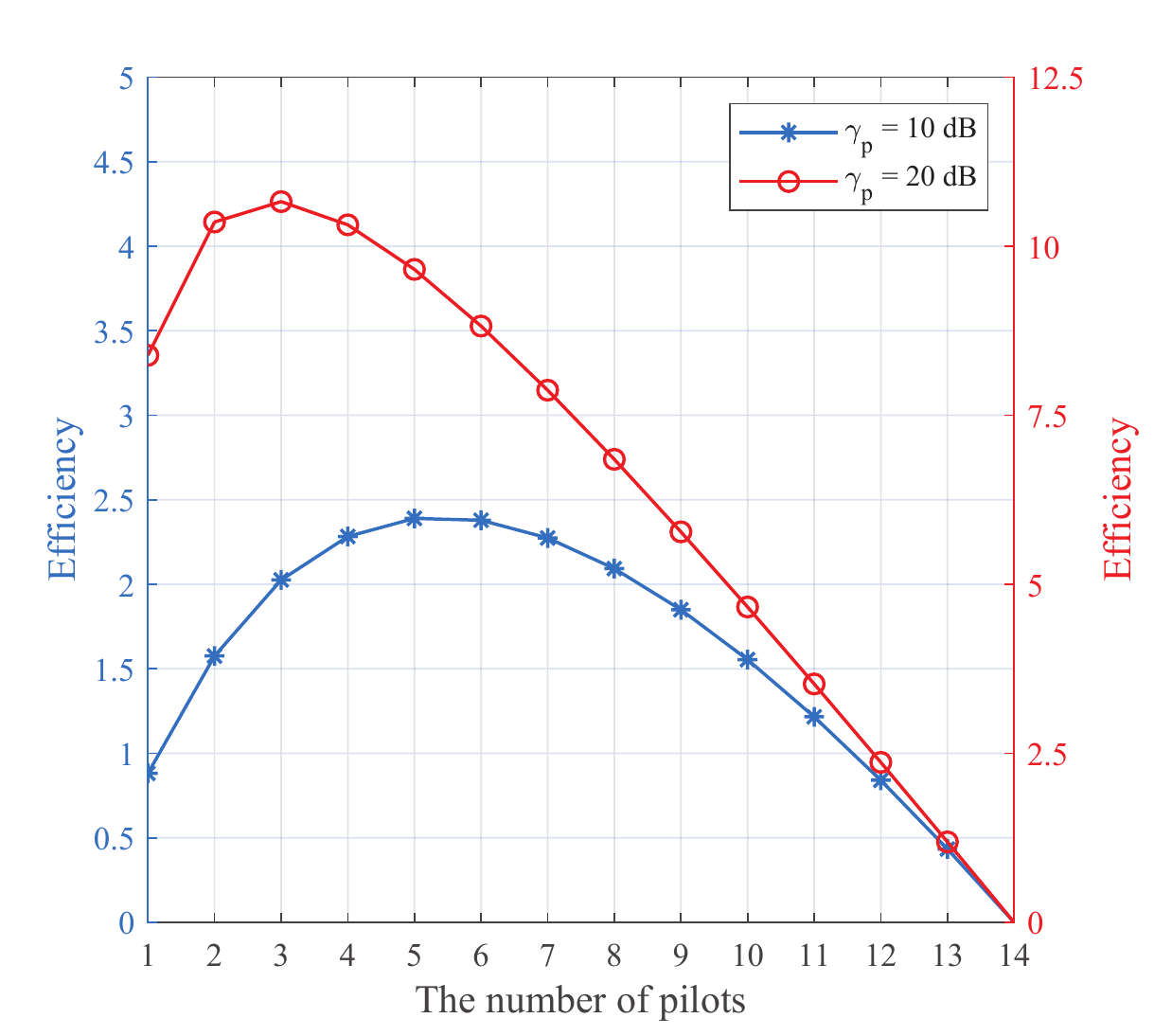}}
\subfigure[Efficiency with SNR.]{\centering 
\includegraphics[width=0.235\textwidth]{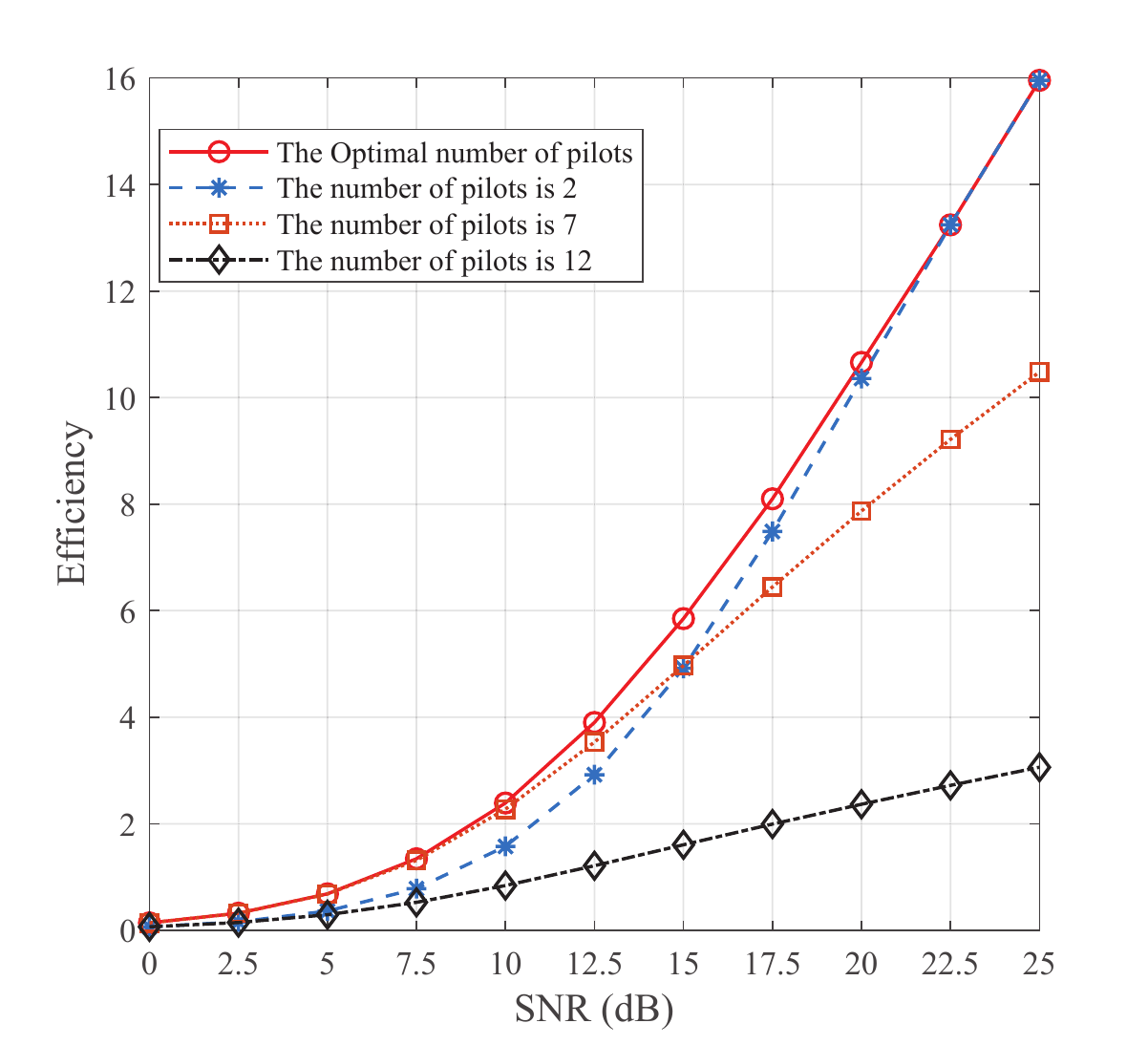}}
\caption{Efficiency of ISAC at range $=$ 100 m.}
\end{figure}

In Fig. 6, the utility of ISAC proposed in \eqref{eqn:utility} is evaluated. In particular, $\gamma_p = 10$ dB, $U_{c,th} = U_{d,th} = 0.2$ and three different settings of $\eta$ are considered. Note that $L_p \leq 2$ and $L_p \geq 12$ are omitted as the utility threshold can not be guaranteed. As shown in this figure, utility increases with $L_p$ increasing when $\eta=0.4$, due to the dominance of range sensing utility. Besides, a peak appearing in the utility curve at $\eta=0.5$ means that there exists an optimal $L_p$ to be selected to achieve a better balance between communications and sensing. Moreover, the decrement of utility with $L_p$ increasing at $\eta=0.6$ means that capacity utility dominates.

\begin{figure}
	\centering
	\includegraphics[width=0.6\linewidth, height=0.2\textwidth]{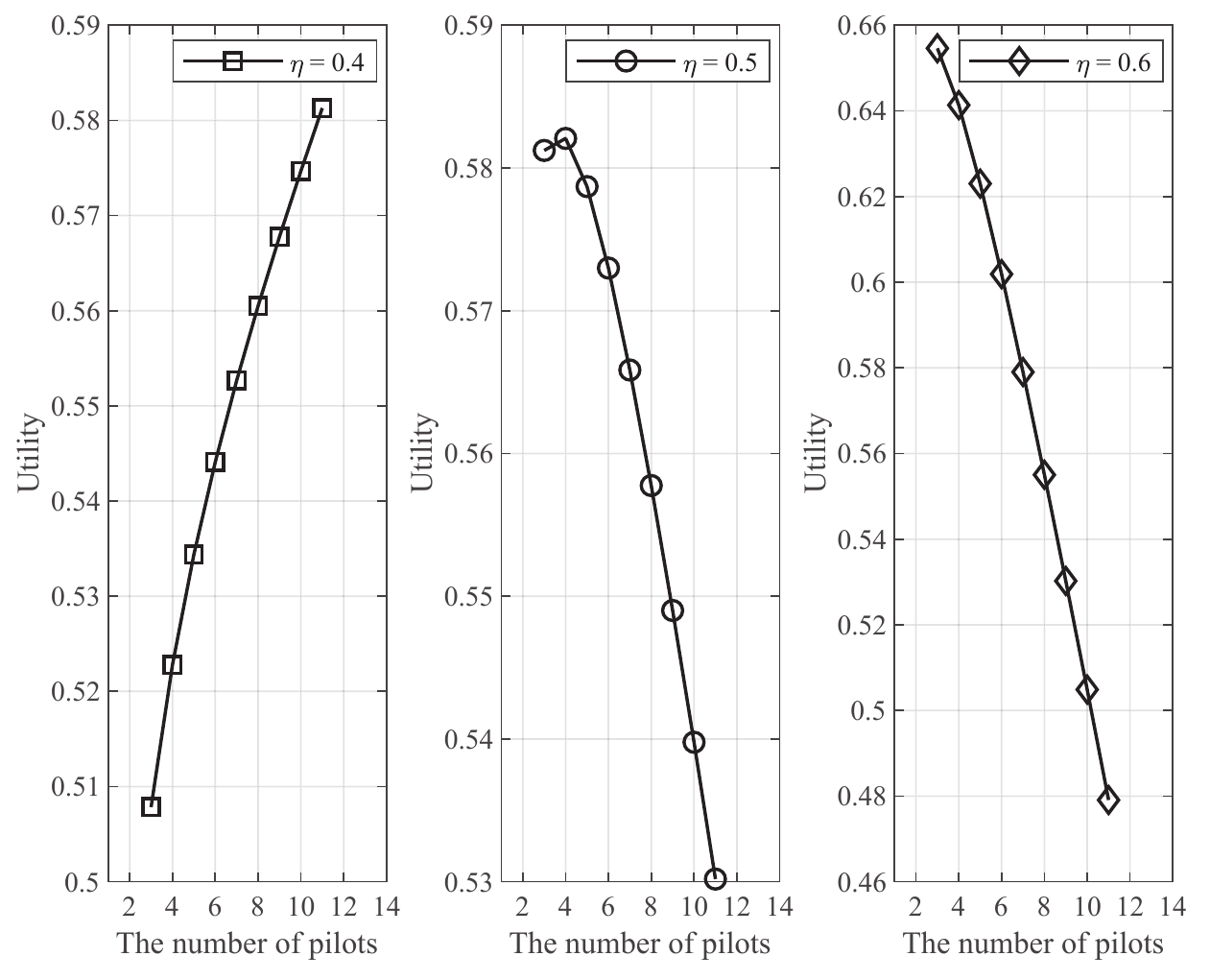}
	\caption{Utility of ISAC at range $=$ 100 m.}
\end{figure}

\section{Conclusions.}
In this paper, the trade-off between the communication and sensing performance in terms of the capacity and range CRB has been investigated. By modeling the communication link and the sensing link as the Rayleigh channel and Rice channel, the closed-form expressions of ergodic capacity with imperfect channel estimation and ergodic range CRB are provided, respectively. Based on the derived results, the trade-off between capacity and range CRB with respect to the number of pilots is clear. Then two metrics named efficiency and utility are proposed to evaluate and optimize the joint performance of ISAC. Finally, simulation results are provided to verify the accuracy of our analysis and show the performance gain by designing the number of pilots.

\bibliographystyle{ACM-Reference-Format}
\bibliography{main.bib}

\end{document}